\documentclass[letterpaper,11pt]{article}

\usepackage{amsthm,amsmath,amssymb}
\usepackage{times,graphicx}
\usepackage[small]{caption}
\usepackage{listings}
\usepackage{url}

\newcommand{\ZZ}{\mathbb{Z}}

\newtheorem{theorem}{Theorem}

\newtheorem{lemma}{Lemma}
\newtheorem{proposition}{Proposition}

\title{Periodicity of identifying codes in strips}
\author{Minghui Jiang\medskip\\
\small
Department of Computer Science,
Utah State University,
Logan, UT 84322-4205, USA\\
\small
\texttt{mjiang@cc.usu.edu}}

\begin{document}

\maketitle

\begin{abstract}
An identifying code in a graph is a subset of vertices 
having a nonempty and distinct intersection with
the closed neighborhood of every vertex.
We prove that the infimum density of any identifying code in $S_k$
(an infinite strip of $k$ rows in the square grid)
can always be achieved by a periodic identifying code with pattern length
at most $2^{4k}$.
Assisted by a compute program implementing Karp's algorithm
for minimum cycle mean,
we find a periodic identifying code in $S_4$ with the minimum density $11/28$,
and a periodic identifying code in $S_5$ with the minimum density $19/50$.

\end{abstract}

\section{Introduction}

For $d \ge 1$,
the \emph{grid} $G_d$ is the (infinite) graph with vertex set
$\ZZ^d$, and with edges between vertices with (Euclidean) distance $1$.
For $k \ge 1$,
the \emph{strip} $S_k$ is the subgraph of $G_2$
induced by the vertex subset $\ZZ\times\ZZ_k$,
where $\ZZ_k := \{0,\ldots,k-1\}$.
For $l \ge 1$,
any subgraph of $S_k$ induced by
$\{j_1,\ldots,j_2\}\times\ZZ_k$
with $j_2-j_1+1 = l$
is called a \emph{bar} with length $l$, or an \emph{$l$-bar}.
For $i\in\ZZ_k$ and $j\in\ZZ$,
the subgraphs of $S_k$ induced by
$\ZZ\times\{i\}$
and
$\{j\}\times\ZZ_k$,
respectively,
are called
a \emph{row} and a \emph{column} of $S_k$;
similarly we also talk about rows and columns of a bar.

Let $G$ be a (finite or infinite) graph.
For any vertex $v \in V(G)$,
the \emph{open neighborhood} $N(v)$
is the subset of vertices adjacent to $v$ in $G$,
and the \emph{closed neighborhood} $N[v]$ is $N(v) \cup \{v\}$.
For $r \ge 0$ and $v \in V(G)$,
the \emph{ball} of radius $r$ centered at $v$,
denoted by $B_r(v)$,
is the set of vertices with distance at most $r$ from $v$ in $G$.
In particular, $B_0(v) = \{ v \}$ and $B_1(v) = N[v]$.
A vertex subset $C \subseteq V(G)$ is an \emph{identifying code} in $G$ if
\begin{enumerate}

\item
for each vertex $v \in V(G)$, $N[v] \cap C \neq \emptyset$,

\item
for each pair of distinct vertices $u,v \in V(G)$,
$N[u] \cap C \neq N[v] \cap C$.

\end{enumerate}

Let $v_0$ be an arbitrary vertex in $G$.
For any $C \subseteq V(G)$,
the \emph{upper density}
and \emph{lower density}
of $C$ in $G$ are, respectively,
$$
\overline{d}(C,G) := \limsup_{r\to\infty} \frac{|C\cap B_r(v_0)|}{|B_r(v_0)|}
\quad\textup{and}\quad
\underline{d}(C,G) := \liminf_{r\to\infty} \frac{|C\cap B_r(v_0)|}{|B_r(v_0)|}.
$$
If these two numbers are equal, then their common value is simply called
the \emph{density} of $C$ in $G$,
$$
d(C,G) := \lim_{r\to\infty} \frac{|C\cap B_r(v_0)|}{|B_r(v_0)|}.
$$
In particular, if $G$ is a finite graph, then $d(C,G)$ always exists, and
$$
d(C,G) = \frac{|C|}{|V(G)|}.
$$

The infimum density of an identifying code in $G$ is customarily defined
as
\begin{equation}\label{eq:def}
d^*(G) := \inf_C \overline{d}(C,G),
\end{equation}
where $C$ ranges over all identifying codes in $G$.
Here the upper density $\overline{d}$ is used because
the density $d$ does not always exist when $G$ is an infinite graph.
We show in this paper that if $G$ is $S_k$, for any $k \ge 1$,
then the infimum density can always be achieved by a
\emph{periodic} identifying code
repeating a bar pattern
whose length $l$ is bounded by a function of $k$,
and hence
\begin{equation}\label{eq:dmin}
d^*(G) = \min_C d(C,G),
\end{equation}
where $C$ ranges over all periodic identifying codes $C$, whose densities exist.

Identifying codes in the square grid $G_2$ and its subgraphs have been
extensively studied.
The concept was introduced by
Karpovsky, Chakrabarty, and Levitin~\cite{KCL98}
for its application in fault diagnosis of multiprocessor systems.
Cohen et~al.~\cite{CHLZ99,CGHLMPZ99} proved an upper bound of
$d^*(G_2) \le 7/20$
by giving two periodic identifying codes achieving this density.
Subsequently,
Ben-Haim and Litsyn~\cite{BL05} proved the matching lower bound of
$d^*(G_2) \ge 7/20$.

\begin{figure}[htbp]
\centering\includegraphics{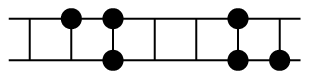}\qquad\includegraphics{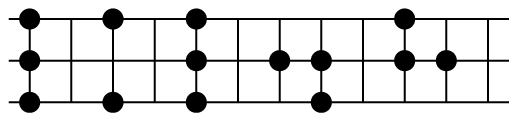}
\caption{Bar patterns of periodic identifying codes
with densities $3/7$ and $7/18$, respectively,
in $S_2$ and $S_3$.}
\label{fig:s2s3}
\end{figure}

Daniel, Gravier, and Moncel~\cite{DGM04} initiated the study of identifying
codes in strips, and proved that $d^*(S_1) = 1/2$ and $d^*(S_2) = 3/7$.
The infimum density of $1/2$ for $S_1$ is achieved
by a periodic identifying code with pattern length $2$,
consisting of every other vertex in $S_1$.
The infimum density of $3/7$ for $S_2$ is achieved
by a periodic identifying code with pattern length $7$,
as illustrated in Figure~\ref{fig:s2s3} left.
Bouznif, Darlay, Moncel, and Preissmann~\cite{BDMP11}
determined that $d^*(S_3) = 7/18$ with a computer-assisted proof,
but left the case of $d^*(S_4)$ open because their program
(based on matrix power)
consumed too much memory.
Refer to Figure~\ref{fig:s2s3} right for the periodic identifying code
with pattern length $12$ that achieves this density.
In his Ph.D. dissertation,
Bouznif~\cite[Section~7.2.2]{Bo12} also proved, among other results,
that the infimum density of a circular strip of height $4$ and length $14n$
is $11/28$,
which implies the upper bound that $d^*(S_4) \le 11/28$.
Recently,
Bouznif, Havet, and Preissmann~\cite{BHP16}
presented a combinatorial proof of $d^*(S_3) = 7/18$
via a sophisticated charging argument.
Extending the previous result of $d^*(G_2) = 7/20$ for the square grid $G_2$,
Bouznif, Havet, and Preissmann~\cite{BHP16} also proved that,
for any $k \ge 1$,
$$
\frac7{20} + \frac1{20k} \le d^*(S_k) \le \min\left\{ \frac25,\,\,
	\frac7{20} + \frac3{10k} \right\}.
$$

\begin{figure}[htbp]
\centering\includegraphics{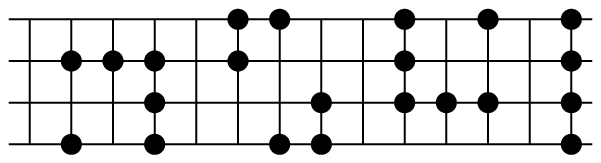}\qquad\includegraphics{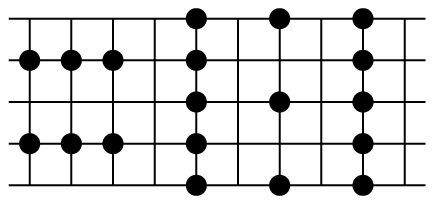}
\caption{Bar patterns of periodic identifying codes
with densities $11/28$ and $19/50$, respectively,
in $S_4$ and $S_5$.}
\label{fig:s4s5}
\end{figure}

In this paper, we prove that $d^*(S_4) = 11/28$ and $d^*(S_5) = 19/50$,
with the assistance of a computer program justified by~\eqref{eq:dmin}.
Refer to Figure~\ref{fig:s4s5}:
the infimum density of $11/28$ for $S_4$ is achieved
by a periodic identifying code with pattern length $14$
(which is exactly the optimal pattern for a circular strip of height $4$ and length $14n$ that Bouznif~\cite[Figure~7.14]{Bo12} previously discovered);
the infimum density of $19/50$ for $S_5$ is achieved
by a periodic identifying code with pattern length $10$.
It is interesting to note that
$d^*(S_4) = 0.3928\ldots$ is greater than both
$d^*(S_3) = 0.3888\ldots$ and
$d^*(S_5) = 0.38$,
in contrast to the strict monotonicity of $d^*(S_1) > d^*(S_2) > d^*(S_3)$,
and to the asymptotic bound of $d^*(S_k) = 7/20 + \Theta(1/k)$.

\section{An alternative interpretation of periodic identifying codes in strips}

Consider any subset $C$ of vertices in a graph $G$.
For any pair of distinct vertices $u,v\in V(G)$,
if
$N[u] \cap C \ne \emptyset$,
$N[v] \cap C \ne \emptyset$,
and
$N[u] \cap N[v] = \emptyset$,
then
$N[u] \cap C \ne N[v] \cap C$.
Thus to verify whether $C$ is an identifying code in $G$,
it suffices to verify that
\begin{enumerate}

\item
for each vertex $v \in V(G)$, $N[v] \cap C \neq \emptyset$,

\item
for each pair of distinct vertices $u,v \in V(G)$
with $N[u] \cap N[v] \ne \emptyset$,
$N[u] \cap C \neq N[v] \cap C$.

\end{enumerate}
Any two vertices with intersecting closed neighborhoods
are within distance $2$ from each other.
In particular, in the strip $S_k$,
such a pair of vertices must be located in some $3$-bar,
and the union of their closed neighborhoods is contained in some $5$-bar.

Let $R$ be any bar with length $l \ge 3$ in the strip $S_k$,
and let $R'$ be the sub-bar with length $l-2$ consisting of
the middle columns (except the first column and the last column) of $R$.
Then the closed neighborhood (in $S_k$) of each vertex in $R'$
is contained in $R$.
We say that a subset $P$ of vertices in $R$ is a \emph{barcode} of $R$,
if $P$ satisfies the two conditions of an identifying code
\emph{locally} for the vertices in $R'$,
that is,
for each vertex $v \in R'$,
$N[v] \cap P \ne \emptyset$,
and for each pair of distinct vertices $u,v\in R'$,
$N[u] \cap P \ne N[v] \cap P$.

Construct an edge-weighted directed graph $H_k$ as follows.
For each barcode $P$ of any $4$-bar,
let $H_k$ have a corresponding vertex.
For each barcode $Q$ of any $5$-bar,
which induces two barcodes $P'$ and $P''$ of $4$-bars
consisting of the first $4$ columns and the last $4$ columns, respectively,
of the $5$-bar,
let $H_k$ have an edge,
directed from the vertex
corresponding to $P'$,
to the vertex
corresponding to $P''$,
then set the weight of this edge to the number of vertices of $Q$
in the last column of the $5$-bar.
Then $H_k$ has at most $2^{4k}$ vertices and at most $2^{5k}$ edges,
where each vertex is incident to at most $2^k$ incoming edges
and at most $2^k$ outgoing edges,
and each edge has an integer weight between $0$ and $k$.

Note that $H_k$ is strongly connected.
Given any barcode $P$ of a $4$-bar,
we can always extend it to a barcode $Q$ of a $5$-bar,
which consists of the vertices of $P$ in the first $4$ columns
and all $k$ vertices in the last column.
This implies that $H_k$ contains an edge progression (a directed path)
of at most $4$ edges
from any vertex $v$, to the vertex $s$ corresponding to the barcode
containing all $4k$ vertices of a $4$-bar.
Symmetrically, there is also an edge progression of at most $4$ edges
from $s$ to any vertex $v$ in $H_k$.
Note also that $H_k$ contains self-loops,
for example, from $s$ to itself.
Indeed if a vertex is connected to itself by an edge,
then the barcode corresponding to this vertex must include
either all or none of the $4$ vertices in each row of the $4$-bar.

Define the \emph{mean weight} of
any finite edge progression in an edge-weighted directed graph as
the total weight of the edges divided by the number of edges
(allowing duplicates).
Define the \emph{minimum cycle mean} of any edge-weighted directed graph
as the minimum mean weight of any cycle in it.
The following proposition is easy to prove:

\begin{proposition}\label{prp}
Each bar pattern of a periodic identifying code in $S_k$
corresponds to a cycle in $H_k$,
and vice versa.
The minimum density of a periodic identifying code in $S_k$
is exactly the minimum cycle mean of $H_k$ divided by $k$.
\end{proposition}

\section{Existence of a periodic identifying code achieving the infimum density}

We next show that the infimum density of an identifying code in $S_k$
can always be achieved by a periodic identifying code with pattern length
bounded by a function of $k$:

\begin{theorem}\label{thm:periodic}
For any $k \ge 1$,
there is a periodic identifying code $C^*$
with pattern length $l \le 2^{4k}$ in $S_k$
satisfying $d(C^*,S_k) = d^*(S_k)$.
\end{theorem}

The proof follows the line of reasoning
of~\cite[Hint for Exercise~3.2]{PA95},
where a similar question on the (lattice) packing density of a convex body in
the plane is considered.
The following lemma is folklore:

\begin{lemma}\label{lem:super}
For any graph $G$,
if $C \subseteq V(G)$ is an identifying code in $G$,
then any subset $D$ with $C \subseteq D \subseteq V(G)$ is also an identifying
code in $G$.
\end{lemma}

\begin{proof}
Since $C \subseteq D$, we have for each vertex $v \in V(G)$,
$$
N[v] \cap C \ne \emptyset \implies N[v] \cap D \ne \emptyset,
$$
and for each pair of distinct vertices $u,v \in V(G)$,
$$
N[u] \cap C \ne N[v] \cap C \implies N[u] \cap D \ne N[v] \cap D.
$$
The lemma then follows by the definition of identifying codes.
\end{proof}

Now we proceed with the proof of Theorem~\ref{thm:periodic}.
Fix $k \ge 1$.
For each $n \ge 4$, we will construct a periodic identifying code $C_n$ in
$S_k$ with pattern length at most $2^{4k}$ such that
\begin{equation}\label{eq:CnSk}
d(C_n, S_k) \le d^*(S_k) + 5 n^{-1}.
\end{equation}

Fix $n \ge 4$.
By the definition of $d^*$ in~\eqref{eq:def},
there exists a sequence of
(not necessarily periodic) identifying codes in $S_k$,
whose upper densities tend arbitrarily close to $d^*(S_k)$.
In particular,
we can find $A_n$ such that
$$
\overline{d}(A_n, S_k) \le d^*(S_k) + n^{-1}.
$$
In the following we will derive $C_n$ from $A_n$.

Partition $S_k$ into an infinite sequence of disjoint $n$-bars $R[i]$,
induced by disjoint vertex subsets
$V[i] := \{ in,\ldots,(i+1)n-1 \}\times\ZZ_k$, $i \in \ZZ$.
Let $A[i] := A_n \cap V[i]$.
Then for all $i\in\ZZ$,
$$
d(A[i], R[i]) = \frac{|A[i]|}{nk},
$$
and hence there exists $j \in \ZZ$ such that
$$
d(A[j], R[j]) \le \overline{d}(A_n, S_k).
$$

Let $Q := \{ nj,\ldots,nj+3 \}\times\ZZ_k$
and $Q' := \{ n(j+1),\ldots,n(j+1)+3 \}\times\ZZ_k$
be the vertices in the first $4$ columns of $R[j]$ and of $R[j+1]$,
respectively.
By Lemma~\ref{lem:super},
$A_n \cup Q \cup Q'$ remains an identifying code in $S_k$.
Let $P := A[j] \cup Q$.
Then simply repeating the pattern $P$ of $R[j]$
results in an identifying code too.
Let $C_n$ denote this periodic identifying code in $S_k$,
with pattern length $n$.
Note that $|P| \le |A[j]| + 4k$.
Thus
$$
d(C_n, S_k)
= \frac{|P|}{nk}
\le \frac{|A[j]| + 4k}{nk}
= d(A[j], R[j]) + 4n^{-1}
\le \overline{d}(A_n, S_k) + 4n^{-1}
\le d^*(S_k) + 5 n^{-1},
$$
as desired in~\eqref{eq:CnSk}.

Recall Proposition~\ref{prp} on the equivalence between
patterns of periodic identifying codes in $S_k$
and cycles in $H_k$.
Since $H_k$ has at most $2^{4k}$ vertices,
any cycle in $H_k$ with more than $2^{4k}$ edges must repeat some vertex.
Then the edges of the cycle can be partitioned to form two shorter cycles,
and the mean weights of the two parts cannot be both greater than that of
the whole.
Correspondingly,
if the length $n$ of the pattern $P$ exceeds $2^{4k}$,
then we can replace $P$ by a shorter pattern with equal or smaller density.
Thus we can assume, without loss of generality, that $C_n$ is a periodic
identifying code with pattern length at most $2^{4k}$.

We have obtained an infinite sequence of periodic identifying codes $C_n$,
for $n \ge 4$,
with pattern length at most $2^{4k}$.
Since $k$ is finite,
the number of distinct periodic identifying codes with pattern length at most
$2^{4k}$ is finite.
So there must exist a periodic identifying code $C^*$
with pattern length at most $2^{4k}$,
such that $C_{n_i}$ is identical to $C^*$
for an infinite sequence of increasing positive integers $n_i$,
$i = 1,2,\dots$.
Then by~\eqref{eq:CnSk} we have $d(C^*, S_k) = d(C_{n_i}, S_k)
\le d^*(S_k) + 5 n_i^{-1}$ for $i = 1,2,\ldots$,
and consequently
$d(C^*, S_k) = d^*(S_k)$.
This completes the proof of Theorem~\ref{thm:periodic}.

\section{An algorithm for finding a minimum-density periodic identifying code}

We briefly review Karp's algorithm~\cite{Ka78}
for computing the minimum cycle mean $\lambda^*$ of
a strongly connected directed graph $G$
with $n$ vertices and $m$ edges,
where each edge $e\in E(G)$ has a weight $f(e)$.
Let $s$ be an arbitrary vertex in $G$.
For each vertex $v \in V(G)$, and each integer $k \ge 0$,
define $F_k(v)$ as the minimum weight of any edge progression of length $k$
from $s$ to $v$;
if no such edge progression exists, then $F_k(v) := \infty$.
Karp proved~\cite[Theorem~1]{Ka78}
(see also~\cite[Problem~24.5]{CLRS09} for a step-by-step sketch of the main
argument in this elegant proof)
that
\begin{equation}\label{eq:lambda}
\lambda^* = \min_{v\in V(G)} \max_{0\le k \le n-1} \frac{F_n(v) - F_k(v)}{n-k}.
\end{equation}

All values of $F_k(v)$ can be computed in $O(nm)$ time by dynamic programming,
with the recurrence
\begin{equation}\label{eq:recurrence}
F_k(v) = \min_{(u,v)\in E(G)} [F_{k-1}(u) + f(u,v)]
\quad\textrm{for } 1 \le k \le n-1,
\end{equation}
and the initial conditions
$$
F_0(s) = 0
\quad\textrm{and}\quad
F_0(v) = \infty
\textrm{ for } v \ne s.
$$
Finally, by~\eqref{eq:lambda},
the value of $\lambda^*$ can be computed in $O(n^2)$ time.
The overall running time is $O(nm)$.

If the actual cycle yielding the minimum cycle mean is desired,
it may be extracted from a minimum-weight edge progression of length $n$
between $s$ and the minimizing vertex $v$ in~\eqref{eq:lambda}.
Such edge progressions for all vertices $v$ can be recorded,
in parallel to the computation of $F_k(v)$ by~\eqref{eq:recurrence},
by recording back references to the vertices $u$ of the minimizing edges
$(u,v)$.

It is easy to verify that Karp's algorithm remains valid for directed graphs
with self-loops.

Now recall Proposition~\ref{prp}.
By applying Karp's algorithm to the directed graph $H_k$ corresponding to the
strip $S_k$, we immediately have the following theorem:

\begin{theorem}
For any $k \ge 1$,
a minimum-density periodic identifying code in $S_k$,
with pattern length at most $2^{4k}$,
can be found by Karp's algorithm in $O(2^{9k})$ time.
\end{theorem}

A straightforward implementation of Karp's algorithm uses $\Theta(n^2)$ space
for all values of $F_k(v)$ and the corresponding back references to $u$.
The space for $F_k(v)$ can be easily reduced to $O(n)$ by running the
dynamic programming algorithm in two passes.
In each pass,
by~\eqref{eq:recurrence},
the values of $\{ F_k(v) \mid v \in V\}$ for each $k \ge 1$
depend only on $\{ F_{k-1}(v) \mid v \in V\}$.
Thus the space for $F_k(v)$ in each round $k$ can be kept for the next round,
while the space for $F_{k-1}(v)$ can be recycled for $F_{k+1}(v)$.
At the end of the first pass, the values of $F_n(v)$ are recorded.
Then in the second pass,
with ready access to $F_n(v)$,
we can compute
$$
\lambda(v) := \max_{0\le k \le n-1} \frac{F_n(v) - F_k(v)}{n - k}
$$
for each $v$ on the fly,
without keeping the values of $F_k(v)$ for all $k$.
Finally, after the second pass, we can compute
$\lambda^* = \min_{v\in V(G)} \lambda(v)$ in $O(n)$ time.
The overall running time remains $O(nm)$.

The $\Theta(n^2)$ space for the back references is harder to reduce.
However, it turns out that in our application,
a rather short cycle yielding the minimum cycle mean
can be found near the end of the minimum-weight edge progression
of length $n$ from $s$ to $v$.
Thus we only need to record back references when $k$ gets close to $n$
in~\eqref{eq:recurrence},
and hence can reduce the total space usage to $O(n + m)$.

Refer to Appendix~\ref{sec:code} for the source code of a computer program
written in the C programming language.
The subscript $k$ of $F_k(v)$ in Karp's algorithm
and the subscript $k$ of $S_k$ have different meanings,
and are denoted by variable \texttt{k} and symbolic constant \texttt{K},
respectively, in the program.
The values of $\{ F_k(v) \mid v \in V\}$
and $\{ F_{k-1}(v) \mid v \in V\}$
are recycled in the two arrays \texttt{d0} and \texttt{d1},
while the values of $\{ F_n(v) \mid v \in V\}$
are recorded in the array \texttt{dn}.
To avoid numerical inaccuracy of floating point arithmetics,
the numerator and denominator of each $\lambda(v)$
are recorded separately, in the two \texttt{int} arrays
\texttt{dd} and \texttt{nk}.
The back references for each vertex $v$ are recorded in the array \texttt{pp}
in \texttt{struct vertex}.

The program has been tested on a laptop computer with modest computing power:
MacBook Air (13-inch, Mid 2011); 1.8 GHz Intel Core i7 processor;
4 GB 1333 MHz DDR3 memory.
With \texttt{K} set to $2$ and $3$, respectively,
the program runs for less than one second,
confirms the known results of
$d^*(S_2) = 3/7$ and $d^*(S_3) = 7/18$,
and rediscovers the minimum-density periodic identifying codes
in Figure~\ref{fig:s2s3}.
With \texttt{K} set to $4$,
the program runs for about three minutes,
determines
$d^*(S_4) = 11/28$,
and finds the minimum-density periodic identifying code
in Figure~\ref{fig:s4s5} left.
With \texttt{K} set to $5$,
the program runs for about $45$ hours,
determines
$d^*(S_5) = 19/50$,
and finds the minimum-density periodic identifying code
in Figure~\ref{fig:s4s5} right.

\appendix

\section{Source code of a computer program}\label{sec:code}

\lstset{language=C, basicstyle=\footnotesize\ttfamily, keywordstyle=\ttfamily,
	showstringspaces=false, tabsize=4 }
\begin{lstlisting}
#include <stdio.h>
#include <stdlib.h>
#include <string.h>
#include <time.h>

#define K	3
#define N	(1 << 4 * K)
#define M	(1 << K)
#define P	(M * K * 2)
#define INF	(N * K + 1)

struct vertex {
	int vv[M];
	int ww[M];
	int pp[P];	/* P instead of N to save space */
	int m;
} *xx[N];

int d0[N], d1[N], dn[N], dd[N], nk[N];
int q[P];

int gcd(int a, int b) {
	return b == 0 ? a : gcd(b, a % b);
}

int count(int b) {
	int cnt = 0;

	while (b > 0) {
		cnt++;
		b &= b - 1;
	}
	return cnt;
}

int h(int j, int i) {
	return j * K + i;
}

int encode(char a[][K], int l) {
	int i, j, b = 0;

	for (j = 0; j < l; j++)
		for (i = 0; i < K; i++)
			if (a[j][i])
				b |= 1 << h(j, i);
	return b;
}

void decode(char a[][K], int l, int b) {
	int i, j;

	for (j = 0; j < l; j++)
		for (i = 0; i < K; i++)
			a[j][i] = (b & 1 << h(j, i)) == 0 ? 0 : 1;
}

int neighborhood(char a[][K], int j, int i) {
	int b = 0;

	if (a[j][i])
		b |= 1 << h(j, i);
	if (a[j - 1][i])
		b |= 1 << h(j - 1, i);
	if (a[j + 1][i])
		b |= 1 << h(j + 1, i);
	if (i - 1 >= 0 && a[j][i - 1])
		b |= 1 << h(j, i - 1);
	if (i + 1 < K && a[j][i + 1])
		b |= 1 << h(j, i + 1);
	return b;
}

int valid_vertex(char a[4][K]) {
	int i;

	for (i = 0; i < K; i++) {
		if (neighborhood(a, 1, i) == 0 || neighborhood(a, 2, i) == 0)
			return 0;
		if (neighborhood(a, 1, i) == neighborhood(a, 2, i))
			return 0;
		if (i >= 1) {
			if (neighborhood(a, 1, i) == neighborhood(a, 1, i - 1))
				return 0;
			if (neighborhood(a, 2, i) == neighborhood(a, 2, i - 1))
				return 0;
			if (neighborhood(a, 1, i) == neighborhood(a, 2, i - 1))
				return 0;
			if (neighborhood(a, 2, i) == neighborhood(a, 1, i - 1))
				return 0;
		}
		if (i >= 2) {
			if (neighborhood(a, 1, i) == neighborhood(a, 1, i - 2))
				return 0;
			if (neighborhood(a, 2, i) == neighborhood(a, 2, i - 2))
				return 0;
		}
	}
	return 1;
}

int valid_edge(char a[5][K]) {
	int i;

	for (i = 0; i < K; i++)
		if (neighborhood(a, 1, i) == neighborhood(a, 3, i))
			return 0;
	return 1;
}

int build_graph() {
	struct vertex *x;
	char a[5][K];
	int n, m, u, v, e, w;

	memset(xx, 0, sizeof(xx));
	n = m = 0;
	for (u = 0; u < N; u++) {
		decode(a, 4, u);
		if (!valid_vertex(a))
			continue;

		x = malloc(sizeof(struct vertex));
		x->m = 0;
		for (e = 0; e < M; e++) {
			decode(a + 4, 1, e);
			if (!valid_vertex(a + 1) || !valid_edge(a))
				continue;

			v = encode(a + 1, 4);
			w = count(e);
			x->vv[x->m] = v;
			x->ww[x->m] = w;
			x->m++;
			m++;
		}
		xx[u] = x;
		n++;
	}
	printf("n = %d  m = %d\n", n, m);
	return n;
}

void karp(int pass, int n) {
	int k, u, v, e, w, percent = K == 4 ? 10 : 1;

	for (v = 0; v < N - 1; v++)
		d0[v] = INF;
	d0[N - 1] = 0;
	if (pass == 2)
		for (v = 0; v < N; v++)
			if (xx[v]) {
				dd[v] = dn[v] - d0[v];
				nk[v] = n - 0;
			}
	for (k = 1; k <= n; k++) {
		if (K >= 4 && k * 100 > n * percent) {
			time_t t;

			time(&t);
			printf("%% %2d: %s", percent, ctime(&t));
			percent += K == 4 ? 10 : 1;
		}
		for (v = 0; v < N; v++)
			d1[v] = INF;
		for (u = 0; u < N; u++)
			if (xx[u])
				for (e = 0; e < xx[u]->m; e++) {
					v = xx[u]->vv[e];
					w = xx[u]->ww[e];
					if (d1[v] > d0[u] + w) {
						d1[v] = d0[u] + w;
						if (pass == 2 && n - k < P)
							xx[v]->pp[n - k] = u;
					}
				}
		for (v = 0; v < N; v++)
			d0[v] = d1[v];
		if (pass == 2 && k < n)
			for (v = 0; v < N; v++)
				if (xx[v] && dn[v] != INF && d0[v] != INF
						&& dd[v] * (n - k) <= (dn[v] - d0[v]) * nk[v]) {
					dd[v] = dn[v] - d0[v];
					nk[v] = n - k;
				}
	}
	if (pass == 1)
		for (v = 0; v < N; v++)
			dn[v] = d0[v];
}

void print_cycle(int j1, int j2) {
	int i, j;

	printf("cycle -%d -%d\n", j1, j2);
	for (i = K - 1; i >= 0; i--) {
		for (j = j1; j < j2; j++)
			printf("%c ", (q[j] & 1 << i) == 0 ? '-' : 'x');
		printf("\n");
	}
}

int main() {
	int n, u, v, w, c, g, k;

	n = build_graph();
	karp(1, n);
	karp(2, n);

	w = INF;
	c = 1;
	v = N - 1;
	for (u = 0; u < N; u++)
		if (xx[u] && w * nk[u] > dd[u] * c) {
			w = dd[u];
			c = nk[u];
			v = u;
		}
	g = gcd(w, c * K);
	printf("w = %d  c = %d  density = %d/%d\n", w, c, w / g, c * K / g);

	for (k = 0; k < P; k++) {
		q[k] = v;
		v = xx[v]->pp[k];
		if (k >= c && q[k] == q[k - c]) {
			print_cycle(k - c, k);
			break;
		}
	}
	return 0;
}
\end{lstlisting}

\end{document}